\documentclass[3p]{elsarticle}

\journal{European J. Combin.}
 \usepackage{afterpage}

\usepackage{graphicx, amssymb, amsthm, amsmath, epsfig, latexsym}
\usepackage[T1]{fontenc}
\usepackage[latin1]{inputenc}
\usepackage[english]{babel}

\usepackage{amscd}
\usepackage{latexsym}
\usepackage{tabularx}
\usepackage{pgf}
\usepackage{pgfcore}
\usepackage{pgfbaseimage}
\usepackage{pgfbaselayers}
\usepackage{pgfbasepatterns}
\usepackage{pgfbaseplot}
\usepackage{pgfbaseshapes}
\usepackage{pgfbasesnakes}
\usepackage{tikz}
\usetikzlibrary{patterns}
\usetikzlibrary{snakes}

\newtheorem{theorem}{Theorem}
\newtheorem{corollary}[theorem]{Corollary}
\newtheorem{definition}[theorem]{Definition}

\newtheorem{proposition}[theorem]{Proposition}
\newtheorem{lemma}[theorem]{Lemma}
\newtheorem{conjecture}[theorem]{Conjecture}

\newcommand{\join}{\bowtie} 
\newcommand{\Sym}{\ominus}

\begin{document}
\begin{frontmatter}
\title{Extremal graphs for the identifying code problem\tnoteref{t1}}
\tnotetext[t1]{This research is
    supported by the ANR Project IDEA {\scriptsize $\bullet$} {ANR-08-EMER-007},  2009-2011.}

\author[LaBRI]{Florent Foucaud}
\author[Grenoble]{Eleonora Guerrini}
\author [LaBRI]{Matja\v{z} Kov\v{s}e}
\author[LaBRI]{Reza Naserasr}
\author[Grenoble] {Aline Parreau}
\author[LaBRI]{Petru Valicov}
\address[LaBRI]{LaBRI - Universit\'e Bordeaux 1 - CNRS,
351 cours de la Lib\'eration,
33405 Talence cedex, France.}
\address[Grenoble]{Institut Fourier 100, rue des Maths, 
BP 74, 38402 St Martin d'H\`eres Cedex, France }

\begin{abstract}
An identifying code of a graph $G$ is a dominating set $C$ 
such that every vertex $x$ of $G$ is distinguished from other vertices
by the set of vertices in $C$ that are at distance at most 1 from $x$.
The problem of finding an identifying code of minimum possible size
turned out to be a challenging problem. It was proved by N.~Bertrand, I.~Charon, O.~Hudry and A.~Lobstein
that if a graph on $n$ vertices with at least
one edge admits an identifying code, then a minimal identifying code
has size at most $n-1$. They introduced classes of graphs whose
smallest identifying code is of size $n-1$. Few conjectures were
formulated to classify the class of all graphs whose minimum
identifying code is of size $n-1$.

In this paper, disproving these conjectures, we classify all finite
graphs for which all but one of the vertices are needed to form an
identifying code.  We classify all infinite graphs needing the whole
set of vertices in any identifying code. New upper bounds in terms of
the number of vertices and the maximum degree of a graph are also
provided.

\end{abstract}


\end{frontmatter}

\section{Introduction}

Given a graph $G$, an identifying code of $G$ is a subset $C$ of vertices
of $G$ such that the subset of $C$ at distance at most 1 from a given
vertex $x$ is nonempty and uniquely determines $x$. Identifying codes
have been widely studied since the introduction of the concept
in~\cite{KCL98}, and have been applied to problems such as
fault-diagnosis in multiprocessor systems~\cite{KCL98}, compact
routing in networks~\cite{LTCS07}, emergency sensor networks in
facilities~\cite{DRSTU03} or the analysis of secondary RNA
structures~\cite{HKSZ06}.

The concept of identifying codes of graphs is related to several 
other concepts, such as locating-dominating sets~\cite{RS84,S88}
for graphs and the well-celebrated theorem of Bondy~\cite{B72}
on set systems.

The purpose of this paper is to classify extremal cases in some previously
known upper bounds for the  minimum size of identifying codes and thus
also improving those upper bounds. We begin by introducing our
terminology.

Unless specifically mentioned $G=(V,E)$ will be a finite simple graph with
$n=|V|$ being the number of vertices. The degree of a vertex $x$ is denoted
$deg(x)$. By $\Delta(G)$ we denote the maximum degree of $G$.

For two vertices $x$ and $y$ of $G$, we denote by $d_G(x,y)$
(or $d(x,y)$ if there is no ambiguity) the distance between $x$ and $y$ in $G$. 
The \emph{ball} of radius~$r$ centered at $x$, denoted $B_r(x)$, is the
set of vertices at distance at most~$r$ of $x$. We note that $x$ belongs
to $B_r(x)$ for every $r$.  A vertex $x$ of $G$ is \emph{universal} if
$B_1(x)=V(G)$. Given a subset $S$ of $V(G)$, we say that a vertex $x$
is $S$-universal if $S\subseteq B_1(x)$. 
 The symmetric difference of two sets $A$ and $B$ is denoted by $A\Sym B$.
 Given a pair of vertices of a graph $G$, we write $\Sym_r(x,y)=B_r(x)\Sym B_r(y)$.
Two vertices $x$ and $y$ are called \emph{twins} in $G$ if $B_1(x)=B_1(y)$. 
A graph is called \emph{twin-free} if it has no pair of twin vertices.
The complement of a graph $G$ is denoted by $\overline G$.
For $r\ge 2$, the \emph{$r^{th}$-power} of $G$, is
the graph $G^r=(V,E')$ with $E'=\{xy~\vert~x,y\in V, d_G(x,y)\le r\}$.  Conversely
if $H^r\cong G$, then we say  $H$ is an \emph{$r$-root} of $G$.
We denote by $G-x$ the graph obtained from $G$ by removing $x$ from $V(G)$
and all edges containing $x$ from $E(G)$. 
For two graphs $G_1=(V_1,E_1)$ and $G_2=(V_2,E_2)$,
$G_1\join G_2$ is the \emph{join graph} of $G_1$ and $G_2$. Its vertex
set is $V_1\cup V_2$ and its edge set is $E_1\cup
E_2\cup\{x_1x_2~\vert~x_1\in V_1, x_2\in V_2\}$. We denote by $K_n$,
the complete graph on $n$ vertices, by $P_n$, the path on $n$
vertices, and by $K_{a,b}$, the complete bipartite graph with
bipartitions of sizes $a$ and $b$.

Given a graph $G$ and an integer $k\ge 2$, a subset $I$ of vertices of
$G$ is called a \emph{$k$-independent set} if for all distinct
vertices $x,y$ of $I$, $d_G(x,y)\ge k$.  A $2$-independent set is
simply an \emph{independent set}.  Given an integer $r\ge 1$, a subset
$S$ of vertices of $G$ is called an \emph{$r$-dominating set} if for
every vertex $x$ of $G$, $B_r(x)\cap S\neq\emptyset$. We say that $S$
\emph{$r$-separates} two vertices $x$ and $y$, if $B_r(x)\cap S\neq
B_r(y)\cap S$.  A subset $S$ of vertices is an \emph{$r$-separating
  set} if it $r$-separates all distinct vertices $x,y$ of $G$. If $S$
is both $r$-dominating and $r$-separating, $S$ is an
\emph{$r$-identifying code}~\cite{KCL98}. If $S$ is $r$-dominating and
$r$-separates vertices of $V(G)\setminus S$, it is called an
\emph{$r$-locating-dominating set}~\cite{RS84}. Given a bipartite
graph $G$ with a partition $V=I \cup A$, a subset $S$ of $A$ is said to
be an \emph{$r$-discriminating code}~\cite{CCCHL08} if $S$
$r$-separates all pairs of distinct vertices of $I$.

In each of the previous concepts when $r=1$, we simply use the name of
the concept without specifying the value of $r$. 

Note that a set $C$ is an $r$-separating set  of $G$ 
(resp. $r$-identifying code) if and
only if it is a separating set (resp. identifying code) of $G^r$. A
graph $G$ admits a separating set (resp. identifying code) if and only
if it is twin-free, as a consequence it admits an $r$-separating set
(resp. $r$-identifying code) if and only if $G^r$ is
twin-free~\cite{CHHL07}.

For a graph $G$, the minimum cardinalities of an $r$-dominating set
and of an $r$-locating-dominating set are commonly denoted by
$\gamma_r(G)$ and $\gamma^{\text{\tiny {LD}} }_r(G)$. If $G^r$ is
twin-free, we denote by $\gamma^{\text{\tiny {ID}} }_r(G)$
(respectively $\gamma^{\text{\tiny {S}} }_r(G)$) the minimum
cardinality of an $r$-identifying code ($r$-separating set) of $G$.
It is clear from the definition that $\gamma^{\text{\tiny {S}} }_r(G)
\leq \gamma^{\text{\tiny {ID}} }_r(G) \leq \gamma^{\text{\tiny {S}}
}_r(G)+1$.

While the exact value of $\gamma^{\text{\tiny {ID}} }$ for some
classes of graphs has been determined~\cite{BCHL04, BCHL05},
finding the value of $\gamma^{\text{\tiny {ID}} }_r(G)$ for a general
graph $G$ is known to be NP-hard for any $r\ge
1$~\cite{CHLZ01,CHL03}.

Upper bounds, in terms of basic graph parameters, have been given for
the minimum sizes of the corresponding sets for most of the previously
defined concepts. In particular it has been
shown that $\gamma^{\text{\tiny {LD}} }_r(G) \leq |V(G)|-1$ and, assuming $G$ is twin-free 
and $G \not\cong \overline{K_n}$,   
$\gamma^{\text{\tiny {ID}} }_r(G) \leq |V(G)|-1$ (see~\cite{S88,GM07,CHL07}). 

For the case of locating-dominating sets, it was proved in~\cite{S88}
that for a connected graph $G$ we have $\gamma^{\text{\tiny {LD}}
}(G) = |V(G)|-1$ if and only if $G$ is either a star or a complete
graph.

In this paper, we do the analogous classification for identifying
codes. In the case of identifying codes, the class of graphs reaching
this bound is a much richer family. Thus we answer, in negative, the
two attempted conjectures for such classification~\cite{S07,
  CCCHL08}. This gives a partial answer to a question posed
in~\cite{CCCHL08}. This is done in Section~\ref{sec:1-id}.

All the previous definitions can easily be extended to infinite graphs. 
Examples of nontrivial infinite graphs for which the whole vertex set
is needed to form an identifying code are given in~\cite{CHL07}. We
classify all such infinite graphs in Section~\ref{sec:infinite}.  In
Section~\ref{BoundBynAndDelta} we introduce new upper bounds for
$\gamma^{\text{\tiny {ID}} }$ in terms of $n$ and $\Delta$. In all
these sections we address the problem of identifying codes only for
$r=1$ . In Section~\ref{r-identifyingCodes} we consider general
$r$-identifying codes.

The next section provides a set of preliminary results.

\section{Preliminary results}

In this section we have put together some basic results necessary 
for our main work. These results could be useful in the study of
identifying codes in general. We start by recalling the following 
theorem.

\begin{theorem}[\cite{B01, GM07}]\label{thm:existence}
  Let $G$ be a twin-free graph on $n$ vertices
  having at least one edge. Then $\gamma^{\text{\tiny {ID}} }(G)\le n-1$.
\end{theorem}

It is shown in~\cite{CHL07} that this bound is tight. In particular it
is shown that for any $t\geq 2$, $\gamma^{\text{\tiny {ID}} }(K_{1,t})=t$.
A stronger result is proved in Section~\ref{BoundBynAndDelta} (see
Lemma~\ref{lemma:takeout_vertex}).

The next lemma is an obvious but a crucial one.

\begin{lemma}\label{lemma:subsetid}
  Let $G $ be a twin-free graph and let $C$ be an
  identifying code of $G$. Then, any set $C'\subseteq V(G)$ such that
  $C\subseteq C'$ is an identifying code of $G$.
\end{lemma}

The next proposition is useful in proving upper bounds on minimum identifying
codes by induction.

\begin{proposition}\label{GeneralizedGoodProposition}
  Let $G$ be a twin-free graph and $S\subseteq V(G)$ such that $G-S$ is
  twin-free. Then 
$\gamma^{\text{\tiny {ID}} }(G)\le \gamma^{\text{\tiny{ID}} }(G-S)+\lvert S\rvert$.
\end{proposition}

\begin{proof}
  Take a minimum code $C_0$ of $G-S$. Consider the vertices of $S$ in
  an arbitrary order $(x_1,\ldots,x_{|S|})$. Using induction we extend
  $C_0$ to a subset $C_i$ of $G$ which identifies the vertices in
  $V_i=V(G)\setminus \{x_{i+1},\ldots,x_{|S|}\}$. To do this, if $C_{i-1}$
  identifies all the vertices of $V_i$, we are done. Otherwise, since
  all the vertices in $V_{i-1}$ are identified, either $B_1(x_i)\cap
  C_{i-1} = B_1(y)\cap C_{i-1}$ for exactly one vertex $y$ in
  $V_{i-1}$, or $x_i$ is not dominated by $C_{i-1}$. In the first case
  $x_i$ and $y$ are separated in $G$ by some vertex, say $u$, so let
  $C_i=C_{i-1}\cup \{u\}$. In the second case,  let $C_i=C_{i-1}\cup
  \{x_i\}$. Now, in both cases, $C_i$ identifies all the vertices of
  $V_i$. At step $|S|$, $C_{|S|}$ is an identifying code of $G$ of
  size at most $\lvert C_0\rvert+\lvert S\rvert\le
  \gamma^{\text{\tiny{ID}} }(G-S)+\lvert S\rvert$.
\end{proof}

We will need the following special case of the previous proposition.

\begin{corollary}\label{GoodCorollary}
 Let $G$ be a connected graph with $\gamma^{\text{\tiny {ID}} }(G)= |V(G)|-1$, $G\ncong K_{1,2}$, then 
there is a vertex $x$ of $G$ such that $G-x$ is still connected and $\gamma^{\text{\tiny {ID}} }(G-x)= |V(G-x)|-1$.
\end{corollary}

\begin{proof}
If $G \cong K_{1,t}$, $t \neq 2$, then any leaf vertex works. Thus, we
may suppose $G\ncong K_{1,t}$. Then by Theorem~\ref{thm:existence},
there is a vertex $x$ of $G$ such that $V(G-x)$ is an identifying code
of $G$ and thus $G-x$ is twin-free and $G-x \ncong \overline{K_n}$. By
Proposition~\ref{GeneralizedGoodProposition}, we have
$\gamma^{\text{\tiny {ID}} }(G-x)\geq \gamma^{\text{\tiny {ID}}
}(G)-1 = |V(G-x)|-1$. Equality holds since otherwise $\gamma^{\text{\tiny {ID}} }(G)= |V(G)|$.
To complete the proof, we show that $x$ can be chosen such that $G-x$ is connected. 
To see this, assume $G-x$ is not connected. 
Since $\gamma^{\text{\tiny {ID}} }(G-x)= |V(G-x)|-1$, except one component, 
every component of $G-x$ is an isolated vertex.
If there are two or more such isolated vertices, then either one of them can be the vertex we want. 
Otherwise there is only one isolated vertex, call it $y$. 
Now if $G-y$ is twin-free, then $y$ is the desired vertex, else 
there is a vertex $x'$ such that $B_1(x')=B_1(x)-y$. 
Then $G-x'$ is connected and twin-free.
\end{proof}

\begin{lemma}\label{lemma:twins_of_the_twins}
 Let $G$ be a twin-free graph and let $v\in V(G)$. Let $x,y$ be a pair of twins in $G-v$. 
If $G-x$ or $G-y$ has a pair of twins, then $v$ must be one of the vertices of the pair.
\end{lemma}
\begin{proof}
Since $v$ separates $x$ and $y$, it is adjacent to one of them (say
$x$) and not to the other. Suppose $z,t$ are twins in $G-x$. Suppose
$z$ is adjacent to $x$ and $t$ is not. If $z\neq v$ then $y$ is also
adjacent to $z$ and, therefore, $t$ is also adjacent to $y$ which
implies $x$ being adjacent to $t$. This contradicts the fact that $x$
separates $z$ and $t$.  The other case is proved similarly.
\end{proof}

\begin{proposition}\label{MinCodeOfJoins}
Let $G_1$ and $G_2$ be twin-free graphs such that for every minimum
separating set $S$ there is an $S$-universal vertex. If $G_1\join G_2$
is twin-free, then we have $\gamma^{\text{\tiny {S}} }(G_1 \join
G_2)=\gamma^{\text{\tiny {S}}} (G_1)+ \gamma^{\text{\tiny {S}}
}(G_2)+1$. Furthermore, if $S$ is a separating set of size
$\gamma^{\text{\tiny {S}}} (G_1)+ \gamma^{\text{\tiny {S}} }(G_2)+1$
of $G_1\join G_2$, then there is an $S$-universal vertex.
\end{proposition}

\begin{proof}
Let $S$ be a minimum separating set of $G_1 \join G_2$. Since vertices of $G_2$ do
not separate any pair of vertices in $G_1$ then $S\cap V(G_1)$ is a separating set
of $G_1$. By the same argument $S\cap V(G_2)$ is a separating set of $G_2$.
Therefore, $|S|\geq \gamma^{\text{\tiny {S}}} (G_1)+ \gamma^{\text{\tiny {S}} }(G_2)$.
But if $|S|= \gamma^{\text{\tiny {S}}} (G_1)+ \gamma^{\text{\tiny {S}} }(G_2)$, then
there is an $[S\cap V(G_1)]$-universal vertex $x$ in $G_1$ and an $[S\cap V(G_2)]$-universal
vertex $y$ in $G_2$. But then, $x$ and $y$ are not separated by $S$.

Given a separating set $S_1$ of $G_1$ and a separating set $S_2$ of
$G_2$, the set $S_1 \cup S_2$ separates all pairs of vertices except
the $S_1$-universal vertex of $G_1$ from the $S_2$-universal vertex of
$G_2$. But since $G_1\join G_2$ is twin-free, we could add one more
vertex to $S_1\cup S_2$ to obtain a separating set of $G_1\join G_2$
of size $\gamma^{\text{\tiny {S}}} (G_1)+ \gamma^{\text{\tiny {S}}}
(G_2)+1$.

For the second part assume $S$ is a separating set of size
$\gamma^{\text{\tiny {S}}} (G_1)+ \gamma^{\text{\tiny {S}} }(G_2)+1$
of $G_1\join G_2$.  Then we have either $|S\cap
V(G_1)|=\gamma^{\text{\tiny {S}}} (G_1)$ or $|S\cap
V(G_2)|=\gamma^{\text{\tiny {S}}} (G_2)$. Without loss of generality
assume the former.  Then there is an $[S\cap V(G_1)]$-universal vertex
$z$ of $G_1$. Since $z$ is also adjacent to all the vertices of $G_2$,
it is an $S$-universal vertex of $G_1\join G_2$.
\end{proof}
In Proposition~\ref{MinCodeOfJoins} if $G_1\ncong K_1$ and $G_2\ncong
K_1$, then $\gamma^{\text{\tiny {ID}} }(G_1\join
G_2)=\gamma^{\text{\tiny {S}} }(G_1 \join G_2)=\gamma^{\text{\tiny
    {S}}} (G_1)+ \gamma^{\text{\tiny {S}} }(G_2)+1$.

The following lemma was discovered in a discussion between the first
author, R.~Klasing and A.~Kosowski. We include a proof for the sake of
completeness.

\begin{lemma}[\cite{Fprivate}]\label{4-independent}
 Let $G$ be a connected twin-free graph, and $I$ be a $4$-independent
 set such that for every vertex $x$ of $I$, the set $V(G)\setminus \{x\}$ is an
 identifying code of $G$. Then $C=V(G)\setminus I$ is an identifying code of
 $G$.
\end{lemma}

\begin{proof}
  Clearly $C$ is a dominating set of $G$.  Let $x,y$ be a pair of
  vertices of $G$. If they both belong to $I$, $C\cap B_1(x) \neq
  C\cap B_1(y)$ because of the distance between $x$ and
  $y$. Otherwise, one of them, say $x$, is in $C$.  If they are not
  separated by $C$, then they must be adjacent. Thus, together they
  could have only one neighbour in $I$, call it $u$.  This is a
  contradiction because $V(G)\setminus \{u\}$ identifies $G$.
\end{proof}

We note that~$4$ is the best possible in the previous lemma. For
example, let $G=P_{4}$ and assume $x$ and $y$ are the two ends of $G$.
It is easy to check that $V(G)\setminus \{x\}$ and $V(G)\setminus
\{y\}$ are both identifying codes of $G$ but $V(G)\setminus \{x,y\}$
is not.

\section{Graphs with $\gamma^{\text{\tiny {ID}}}(G)=|V(G)|-1$}\label{sec:1-id}

In this section we classify all graphs $G$ for which 
$\gamma^{\text{\tiny {ID}} }(G)=|V(G)|-1$. As already mentioned, stars are examples of
such graphs. To classify the rest we show that special powers of paths are the basic
examples of such graphs.
Then we show that any other example is mainly obtained from the join of
some basic elements.

\begin{definition}
For an integer $k\geq 1$, let $A_k=(V_k, E_k)$ be the graph
with vertex set $V_k=\{x_1,\ldots,x_{2k}\}$ and edge set
$E_k=\{x_{i}x_{j} ~\big\vert ~\vert i-j\vert \leq k-1\}$.
\end{definition}

    \begin{figure}[ht]
        \begin{center}
\scalebox{0.9}{\begin{tikzpicture}[join=bevel,inner sep=0.5mm,]
  \node (x_kplus1) at (90bp,200bp) [draw, circle, fill=black] {};
  \node  at (87bp,210bp) {$x_{k+1}$};
  \node (x_kplus2) at (130bp,200bp) [draw, circle, fill=black] {};
  \node  at (127bp,210bp) {$x_{k+2}$};
  \node (x_kplus3) at (170bp,200bp) [draw, circle, fill=black] {};
  \node  at (167bp,210bp) {$x_{k+3}$}; 
  \node[scale=2] at (200bp,200bp) {$...$}; 
  \node (x_2kminus1) at (230bp,200bp) [draw, circle, fill=black] {};
  \node  at (227bp,210bp) {$x_{2k-1}$};
  \node (x_2k) at (270bp,200bp) [draw, circle, fill=black] {};
  \node  at (267bp,210bp) {$x_{2k}$};
  \node (x1) at (90bp,110bp) [draw, circle, fill=black] {};
  \node  at (87bp,100bp) {$x_1$}; 
  \node (x2) at (130bp,110bp) [draw, circle, fill=black] {};
  \node  at (127bp,100bp) {$x_2$};  
  \node (x3) at (170bp,110bp) [draw, circle, fill=black] {};
  \node  at (167bp,100bp) {$x_3$};  
  \node[scale=2] at (200bp,110bp) {$...$};  
  \node (x_kminus1) at (230bp,110bp) [draw, circle, fill=black] {};
  \node  at (227bp,100bp) {$x_{k-1}$};
  \node (x_k) at (270bp,110bp) [draw, circle, fill=black] {};
  \node  at (267bp,100bp) {$x_k$};  
  \draw [-] (x_kplus1) -- node[above,sloped] {} (x_k);
  \draw [-] (x_kplus1) -- node[above,sloped] {} (x2);
  \draw [-] (x_kplus1) -- node[above,sloped] {} (x3);
  \draw [-] (x_kplus1) -- node[above,sloped] {} (x_kminus1);
  \draw [-] (x_kplus2) -- node[above,sloped] {} (x_kminus1);
  \draw [-] (x_kplus2) -- node[above,sloped] {} (x_k);
  \draw [-] (x_kplus2) -- node[above,sloped] {} (x3);
  \draw [-] (x_kplus3) -- node[above,sloped] {} (x_kminus1);
  \draw [-] (x_kplus3) -- node[above,sloped] {} (x_k);
  \draw [-] (x_2kminus1) -- node[above,sloped] {} (x_k);
  \node at (360bp,200bp) {Clique on $\{x_{k+1},...,x_{2k}\}$};
\draw[line width=1pt, draw = black] (175bp,203bp) ellipse(120bp and 25bp) node{};
  \node at (355bp,110bp) {Clique on $\{x_1,...,x_k\}$};
  \draw[line width=1pt, draw = black] (175bp,108bp) ellipse(120bp and 25bp) node{};  
  \end{tikzpicture}}
  \end{center}
  \caption{The graph $A_{k}$ which needs $|V(A_k)|-1$ vertices for any identifying code}
\label{fig:A_k}
\end{figure}
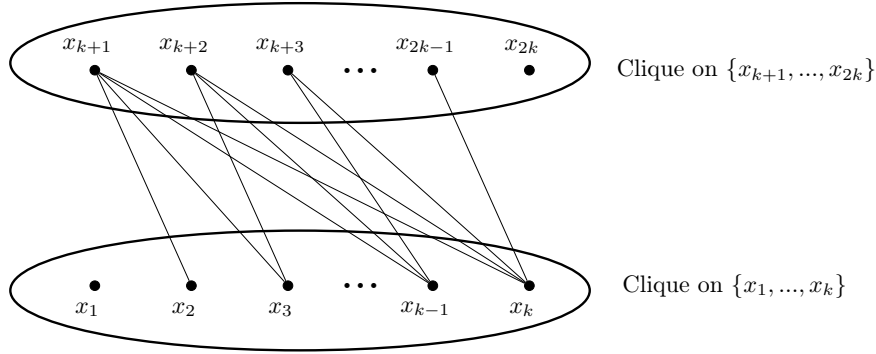

An illustration of graph $A_k$ is given in Figure~\ref{fig:A_k}. We note that for $k\geq 2$ we have $A_{k}=P_{2k}^{k-1}$ and
$A_{1}=\overline{K_{2}}$. It is also easy to check that the only nontrivial 
automorphism of $A_k$ is the mapping $x_{i}\rightarrow x_{2k+1-i}$. It is not hard to
observe that $A_{k}$ is twin-free, $\Delta(A_{k})=2k-2$ and that $A_{k}$ and 
$\overline{A_{k}}$ are connected if $k\geq 2$.

\begin{proposition}\label{prop:A2k}
For $k\geq 1$, we have: $\gamma^{\text{\tiny {S}} }(A_k)=2k-1$
with $B_1(x_k)$ and $B_1(x_{k+1})$ being the only separating sets of size $2k-1$
of $A_k$. Furthermore, if $k \geq 2$, $\gamma^{\text{\tiny {ID}} }(A_k)=2k-1$.
\end{proposition}

\begin{proof}
Let $S$ be a separating set of $A_{k}$. For $i<k$, we have
$\Sym(x_i,x_{i+1})=\{x_{i+k}\}$ and for $k< i \leq 2k-1$, we have
$\Sym(x_i,x_{i+1})=\{x_{i-k+1}\}$. Thus, $\{x_{2},\ldots,x_{2k-1}\}
\subset S$.  But to separate $x_k$ and $x_{k+1}$, we must add $x_1$ or
$x_{2k}$.  It is now easy to see that $V_k \setminus
\{x_1\}=B_1(x_{k+1})$ and $V_k\setminus \{x_{2k}\}=B_1(x_k)$, each is
a separating set of size $2k-1$.  If $k\geq 2$, then they both
dominate $A_{k}$ and therefore are also identifying codes.
\end{proof}

In the previous proof in fact we have also proved that:

\begin{corollary}\label{coro:A2k}
 For $k\geq 1$ every minimum separating set $S$ of $A_k$ has a
 $S$-universal vertex.
\end{corollary}

Let $\mathcal A$ be the closure of $\{A_i ~|~ i=1,2,\ldots \}$ with
 respect to operation $\join$. It is shown below that elements of
$\mathcal A$  are also extremal graphs with respect to both
separating sets and identifying codes. 

\begin{proposition}\label{IDCodeofJoins}
 For every graph $G \in \mathcal
 A$, we have $\gamma^{\text{\tiny{S}} }(G)=|V(G)|-1$. Furthermore,
 every minimum separating set $S$ of $G$ has an $S$-universal vertex.
\end{proposition}

\begin{proof}
The proposition is true for basic elements of $\mathcal A$ by
Proposition~\ref{prop:A2k} and by Corollary~\ref{coro:A2k}. For a
general element $G=G_1\join G_2$ it is true by
Proposition~\ref{MinCodeOfJoins} and by induction.
\end{proof}

\begin{corollary}
 If $G\in \mathcal A$ and $G\ncong A_1$, then $\gamma^{\text{\tiny {ID}} }(G)=|V(G)|-1$.
\end{corollary}

Further examples of graphs extremal with respect to separating sets and identifying codes
can be obtained by adding a universal vertex to each of the graphs in $\mathcal A$, as we
prove below.

\begin{proposition}
 For every graph $G$ in $\mathcal A \join K_1$ we have
$\gamma^{\text{\tiny{ID}} }(G)=\gamma^{\text{\tiny {S}}}(G)=|V(G)|-1$.
\end{proposition}

\begin{proof}
 Assume $G=G_1\join K_1$ with $G_1\in \mathcal A$,
and assume $u$ is the vertex corresponding to $K_1$.
Suppose $S$ is a minimum separating set of $G$. We first note that since $S \cap V(G_1)$
is a separating set of $G_1$, we have $|S\cap V(G_1)|\geq |V(G_1)|-1$. 
But if
$|S\cap V(G_1)|= |V(G_1)|-1$, then by
Proposition~\ref{IDCodeofJoins}, there is an $[S\cap V(G_1)]$-universal vertex $y$ of $G_1$.
Then $y$ is not separated from $x$. Thus $|S\cap V(G_1)|= |V(G_1)|$ and therefore 
$S= V(G_1)$. It is easy to check that $S$ is also an identifying code.
\end{proof}

It was proved in~\cite{CHL07} that $\gamma^{\text{\tiny{ID}}}(K_n\setminus M)=n-1$
where $K_n\setminus M$ is the complete graph minus a maximal matching.
We note that this graph, for even values of $n$, is the join of $\frac n 2$ disjoint copies of $A_1$, thus it belongs to $\mathcal A$. For odd values of $n$, it is built from the previous graph by adding a universal vertex.

So far we have seen that $\gamma^{\text{\tiny{ID}} }(G)=|V(G)|-1$ for
$G\in \{K_{1,t}~|~t\geq 2\}\cup \mathcal A \cup (\mathcal A \join K_1)$, $G\not \cong A_1$.
We also know that $\gamma^{\text{\tiny{ID}} }(\overline {K}_n)=n$.
More examples of graphs with $\gamma^{\text{\tiny{ID}} }(G)=|V(G)|-1$ 
can be obtained by adding isolated vertices.
In the next theorem we show that for any other twin-free graph $G$ we have 
$\gamma^{\text{\tiny{ID}} }(G)\leq |V(G)|-2$.

\begin{theorem}\label{thm:classification}
Given a connected graph $G$, we have $\gamma^{\text{\tiny{ID}} }(G)=|V(G)|-1$ if and only if 
 $G \in \{K_{1,t}~|~t\geq 2\}\cup \mathcal A \cup (\mathcal A \join K_1)$ and $G\not \cong A_1$.
\end{theorem}

\begin{proof}
The ``if'' part of the theorem is already proved. The proof of the 
``only if'' part is based on induction
on the number of vertices of $G$.
For graphs on at most 4 vertices this is easy to check.
Assume the claim is true for graphs on at most $n-1$ vertices and, by 
contradiction, let $G$ be a twin-free graph
on $n\geq 5$ vertices such that $\gamma^{\text{\tiny{ID}} }(G)=n-1$ and 
$G \notin \{K_{1,t}~|~t\geq 2\}\cup \mathcal A \cup (\mathcal A \join K_1)$.

By Corollary~\ref{GoodCorollary} there is a vertex $x\in V(G)$ such that
$G-x$ is connected and $\gamma^{\text{\tiny{ID}} }(G-x)=|V(G-x)|-1$. By the induction hypothesis we have
$G-x \in \{K_{1,t}~|~t\geq 2\}\cup \mathcal A \cup (\mathcal A \join
K_1)$.
Depending on which one of these 3 sets $G-x$ belongs to, we will have 3
cases.

{\bf Case 1}, $G-x \in \{K_{1,t}~|~t\geq 2\}$.
In this case we consider a minimum identifying code $C$ of $G-x$. If $C$ does not already
identify $x$ then either $deg(x) \leq 3$ or $deg(x) \geq n-2$. 
We leave it to the reader to check that in each of these cases, there is 
an identifying code of size $n-2$.

{\bf Case 2}, $G-x \in \mathcal A $. We consider two subcases. Either $G-x \cong
A_k$ for some $k$ or $G-x= G_1\join G_2$, with $G_{1}, G_{2} \in \mathcal A$. 

\begin{itemize}
 \item[(1)] $G-x \cong A_k$, for some $k\geq 2$. If $x$ is adjacent to all
the vertices of $G-x$, then 
$G\in \mathcal A\join K_1$ and we are done. Otherwise there is a pair of consecutive vertices 
of $A_k$, say $x_i$ and $x_{i+1}$, such that one is adjacent to $x$ and the other is not. 
By the symmetry of $A_k$ we may assume $i\leq k$. We claim that $C=V(G) \setminus \{x_1, x\}$
or $C'=V(G)\setminus \{x_{2k}, x\}$ is an identifying code of $G$. This would 
contradict our assumption.
We first consider $C$ and note that $C\cap V(A_k)$ is an identifying code of $A_k$.
If $x$ is also separated from all the vertices of $G-x$ then we are done. Otherwise there will
be two possibilities.

First we consider the possibility: $x$ is not adjacent to $x_i$ and
adjacent to $x_{i+1}$. In this case each vertex $x_j$, $j> i+k$,
is separated from $x$ by $x_{i+1}$ and each vertex $x_j$, $j<i+k$, is separated from $x$
by $x_i$. Thus $x$ is not separated from $x_{i+k}$. In the other
possibility, $x$ is adjacent to $x_i$ and 
not adjacent to $x_{i+1}$. A similar argument implies that $x$ is separated from every vertex but 
$x_1$. In either of these two possibilities, $C'$ would be an identifying code.

\item[(2)] $G-x\cong G_1 \join G_2 $ with $G_1, G_2 \in \mathcal A$. If $x$ is adjacent to 
all the vertices of $G-x$, then $G \in \mathcal A \join K_1$ and we are done. Thus there is
a vertex, say $y$, that is not adjacent to $x$. Without loss of
generality, we can assume $y \in V(G_{1})$.
Let $C_1$ be an identifying code of size $\gamma^{\text{\tiny{ID}} }(G_1)=|V(G_1)|-1$
of $G_1$ which contains $y$. The existence of such an identifying code becomes
apparent from
the proof of Proposition~\ref{IDCodeofJoins}. 
Then $C=C_1\cup V(G_2)$ is an identifying code of $G_1\join G_2$ of size 
$|V(G_1\join G_2)| -1=|V(G)|-2$. Thus $C$ does not separate a vertex of $G_1\join G_2$ from $x$. 
Call this vertex $z$. 
Since $y\in C$, $z$ is not adjacent to $y$, hence $z\in V(G_1)$.
Therefore, $z$ is adjacent to all the vertices of $G_2$. So $x$ should also be
adjacent to all the vertices of $G_2$. Thus we have $G=(G_1+x)\join G_2$ and any minimum 
identifying code of $G_1+x$ together with all vertices of $G_2$ would 
form an identifying code of $G$. This proves that $\gamma^{\text{\tiny{ID}} }(G_1+x)=|V(G_1+x)|-1$. 
Since $G_1+x$ has less vertices than $G$, by induction hypothesis, we have 
$G_1 +x \in \{K_{1,t}~|~t\geq 2\}\cup \mathcal A \cup (\mathcal A \join
K_1)$ and $G\not \cong A_1$.
Since $G_1 \in \mathcal A$, and since $x$ is not adjacent to a vertex of
$G_1$, we should have
$G_1+x\in \mathcal A$ but all graphs in $\mathcal A$ have an even number of vertices and this is not possible.

\end{itemize}

{\bf Case 3}, $G-x \in \mathcal A \join K_1$. Suppose 
$G-x \cong A_{i_1}\join A_{i_2} \join \ldots \join A_{i_j}\join K_1$ and let $u$
be the vertex corresponding to $K_1$.

If $x$ is also adjacent to $u$, then $u$ is a universal vertex of $G$ and 
$G-u$ is also twin-free. In this case we apply the induction on $G-u$:
by Proposition~\ref{GeneralizedGoodProposition},
$\gamma^{\text{\tiny{ID}} }(G-u)=|V(G-u)|-1$
and by induction hypothesis 
$G-u \in \{K_{1,t}~|~t\geq 2\}\cup \mathcal A \cup (\mathcal A \join
K_1)$. But if $ G-u \in \{K_{1,t}~|~t\geq 2\}\cup (\mathcal A \join K_1)$, there will be two 
universal vertices, and therefore twins.
Thus $G-u \in \mathcal A $ and $G \in \mathcal A \join K_1$.

We now assume $x$ is not adjacent to $u$ and we repeat the argument with $G-u$ 
if it is twin-free. In this case if 
$G-u \in \{K_{1,t}~|~t\geq 2\}\cup \mathcal A $, we apply Case 1 or Case 2.
If $G-u\in \mathcal A \join K_1$ with $u'$ being the vertex of $K_1$,
then $u$ and $u'$ induce an isomorphic copy of $A_1$ and $G\in \mathcal A$. 

If $G-u$ is not twin-free then, by Lemma \ref{lemma:twins_of_the_twins},
$x$ must be one of the
twin vertices. Let $x'$ be its twin and suppose $x' \in V(A_{i_1})$ with
$V(A_{i_1})=\{z_1,z_2, \ldots, z_{2k}\}$. Without loss of generality we may assume $x'=z_l$
with $l\leq k$. If $ l\geq 2$, then we claim $C=V(G)\setminus \{z_l, z_{2k}\}$ is an 
identifying code of $G$ which is a contradiction. To prove our claim 
notice first that vertices of $A_{i_2}\join \cdots \join A_{i_j}$ are already identified
from each other and from the other vertices. Now each pair of
vertices of $A_{i_1}$ is separated by a vertex in $V(A_{i_1})\cap C$ except $z_{l+k-1}$ and
$z_{l+k}$ which are separated by $x$. The vertex $x$ is also separated from all the
other vertices by $u$. It remains to show that $u$ is separated from vertices of $A_{i_1}$.
It is separated from vertices in $\{z_1,\ldots, z_{l+k-1}\}$ by $x$ and from
$\{z_{k+1}, \ldots, z_{2k}\}$ by $z_1$ ($l\geq 2)$. 
Thus $x'=x_1$ and now it is easy to see that the subgraph induced by $V(A_{i_1})$, $u$ and
$x$ is isomorphic to $A_{i_1+1}$ and, therefore,
$G\cong A_{i_1+1}\join A_{i_2} \join \ldots \join A_{i_j}$.
\end{proof}

Since every graph in $\{K_{1,t}~|~t\geq 2\}\cup
\mathcal{A}\cup(\mathcal{A}\join K_1)$ has maximum degree $n-2$, we
have:

\begin{corollary}
  Let $G$ be a twin-free connected graph on $n\geq 3$ vertices and maximum degree
  $\Delta\le n-3$. Then $\gamma^{\text{\tiny {ID}} }(G)\le n-2$.
\end{corollary}

\section{Infinite graphs}\label{sec:infinite}

It is shown in~\cite{CHL07} that Theorem~\ref{thm:existence} does not
have a direct extension to the family of infinite graphs. In other
words, there are nontrivial examples of twin-free infinite graphs
requiring the whole vertex set for any identifying code. The basic
example of such infinite graphs, originally defined in~\cite{CHL07},
is given below. In this section, we classify all such infinite
graphs. This strengthens a theorem of~\cite{GM07}, which claims that
there are no such infinite graphs in which all vertices have finite degrees.

\begin{definition}
 Let $X=\{\ldots, x_{-1},x_0,x_1,\ldots \}$ and $Y=\{\ldots, y_{-1},y_0,y_1,\ldots \}$.
$A_{\infty}=(X\cup Y, E)$ is the graph on $X\cup Y$ having edge set
$E=\{ x_ix_j~|~i\neq j\}\cup \{ y_iy_j ~|~i\neq j\}\cup \{ x_iy_j~|~i<j\}$.
\end{definition}

See Figure~\ref{FigureA_infty} for an illustration.

    \begin{figure}[ht]
        \begin{center}
\scalebox{0.9}{\begin{tikzpicture}[join=bevel,inner sep=0.5mm,]
  \node[scale=2] (emptyX1)  at (47bp,200bp) {$...$};
  \node (x_minus2) at (90bp,200bp) [draw, circle, fill=black] {};
  \node  at (87bp,210bp) {$x_{-2}$};
  \node (x_minus1) at (130bp,200bp) [draw, circle, fill=black] {};
  \node  at (127bp,210bp) {$x_{-1}$};
  \node (x_0) at (170bp,200bp) [draw, circle, fill=black] {};
  \node  at (167bp,210bp) {$x_0$}; 
  \node (x_1) at (210bp,200bp) [draw, circle, fill=black] {};
  \node  at (207bp,210bp) {$x_1$};
  \node (x_2) at (250bp,200bp) [draw, circle, fill=black] {};
  \node  at (247bp,210bp) {$x_2$};
  \node[scale=2] (emptyX2) at (287bp,200bp) {$...$};

  \node[scale=2] (emptyY1) at (47bp,110bp) {$...$};
  \node (y_minus2) at (90bp,110bp) [draw, circle, fill=black] {};
  \node  at (87bp,100bp) {$y_{-2}$}; 
  \node (y_minus1) at (130bp,110bp) [draw, circle, fill=black] {};
  \node  at (127bp,100bp) {$y_{-1}$};  
  \node (y_0) at (170bp,110bp) [draw, circle, fill=black] {};
  \node  at (167bp,100bp) {$y_0$};  
  \node (y_1) at (210bp,110bp) [draw, circle, fill=black] {};
  \node  at (207bp,100bp) {$y_1$};
  \node (y_2) at (250bp,110bp) [draw, circle, fill=black] {};
  \node  at (247bp,100bp) {$y_2$}; 
  \node[scale=2] (emptyY2)  at (287bp,110bp) {$...$};

  \draw [-] (emptyX1) -- node[above,sloped] {} (y_minus2);
  \draw [-] (emptyX1) -- node[above,sloped] {} (y_minus1);
  \draw [-] (emptyX1) -- node[above,sloped] {} (y_0);
  \draw [-] (emptyX1) -- node[above,sloped] {} (y_1);
  \draw [-] (emptyX1) -- node[above,sloped] {} (y_2);
  \draw [-] (emptyX1) -- node[above,sloped] {} (emptyY2);
  \draw [-] (x_minus2) -- node[above,sloped] {} (y_2);
  \draw [-] (x_minus2) -- node[above,sloped] {} (y_minus1);
  \draw [-] (x_minus2) -- node[above,sloped] {} (y_0);
  \draw [-] (x_minus2) -- node[above,sloped] {} (y_1);
  \draw [-] (x_minus2) -- node[above,sloped] {} (emptyY2);
  \draw [-] (x_minus1) -- node[above,sloped] {} (y_1);
  \draw [-] (x_minus1) -- node[above,sloped] {} (y_2);
  \draw [-] (x_minus1) -- node[above,sloped] {} (y_0);
  \draw [-] (x_minus1) -- node[above,sloped] {} (emptyY2);
  \draw [-] (x_0) -- node[above,sloped] {} (y_1);
  \draw [-] (x_0) -- node[above,sloped] {} (y_2);
  \draw [-] (x_0) -- node[above,sloped] {} (emptyY2);
  \draw [-] (x_1) -- node[above,sloped] {} (y_2);
  \draw [-] (x_1) -- node[above,sloped] {} (emptyY2);
  \draw [-] (x_2) -- node[above,sloped] {} (emptyY2);
  \node at (405bp,200bp) {Infinite clique on $X$};
\draw[line width=1pt, draw = black] (167bp,203bp) ellipse(170bp and 25bp) node{};
  \node at (405bp,110bp) {Infinite clique on $Y$};
  \draw[line width=1pt, draw = black] (167bp,108bp) ellipse(170bp and 25bp) node{};  
  \end{tikzpicture}}
  \end{center}
  \caption{The graph $A_{\infty}$ which needs all its vertices for any identifying code}
\label{FigureA_infty}
\end{figure}
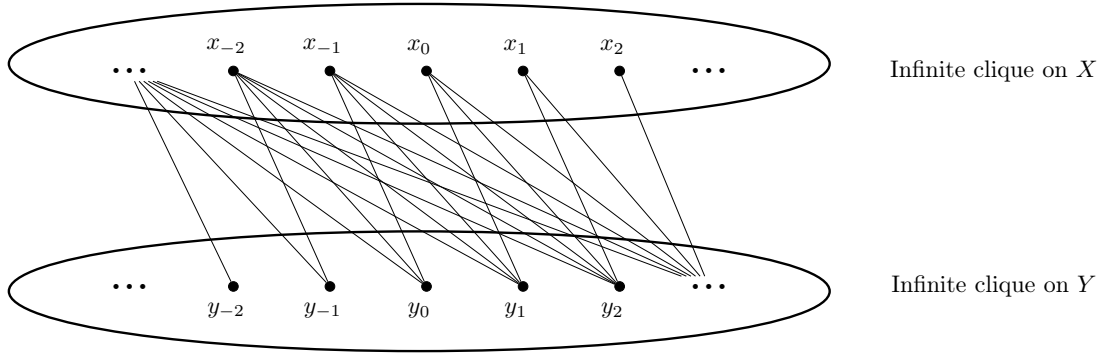

It is shown in~\cite{CHL07} that the only separating set of
$A_{\infty}$ is $V(A_{\infty})$.  One should note that the graph
induced by $\{y_1,y_2, \ldots, y_k, x_1, x_2 \ldots, x_k \}$ is
isomorphic to the graph $A_k$.

Before introducing our theorem let us see again why every separating set 
of $A_{\infty}$ needs the whole vertex set: for every $i$,
$x_{i}$ and $x_{i+1}$ are only separated by $y_{i+1}$, while $y_{i}$
and $y_{i+1}$ are separated only by $x_i$.

This property would still hold if we add a new vertex which is adjacent either to all
vertices in $X$ (similarly in $Y$) or to none. This leads to the
following family:

Let $H$ be a finite or infinite simple graph with a perfect matching $\rho$, that is a
mapping $x \rightarrow \rho(x)$ of $V(H)$ to itself such that $\rho^2(x)=x$
and $x\rho(x)$ is an edge of $H$. We define $\Psi(H,\rho)$ to be the graph built as follows:
for every vertex $x$ of $H$ we assign $\Phi(x)=\{ \ldots x_{-1}, x_0, x_1,
\ldots \}$.
The vertex set of $\Psi(H,\rho)$ is $\displaystyle {\bigcup_{x\in V(H)} \Phi(x)}$.
For each edge $x\rho(x)$ of $H$ we build a copy of $A_{\infty}$ on
$\Phi(x)\cup \Phi(\rho(x))$
and for every other edge $xy$ of $H$ we join every vertex in $\Phi(x)$ to every
vertex in $\Phi(y)$. An example of such construction is illustrated in 
Figure~\ref{ConstructionOfPsi}.

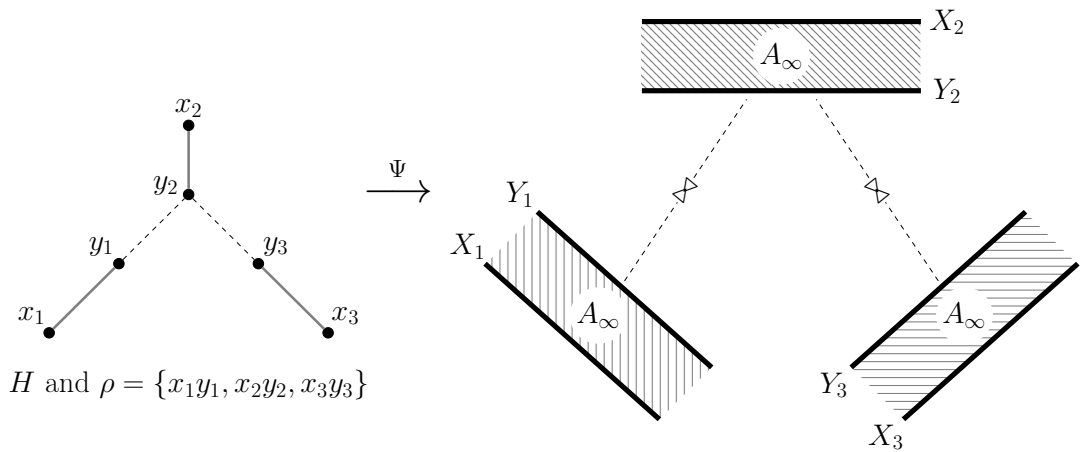
\begin{figure}[ht]
        \begin{center}
\scalebox{0.65}{\begin{tikzpicture}[join=bevel,inner sep=0.5mm,]

  \node (x_1) at (0bp,10bp) [draw, scale=1.5pt, circle, fill=black] {};
  \node  at (-10bp,20bp) {\LARGE $x_1$}; 
  \node (y_1) at (40bp,50bp) [draw, scale=1.5pt, circle, fill=black] {};
  \node  at (30bp,60bp) {\LARGE $y_1$};  
  \node (y_2) at (80bp,90bp) [draw, scale=1.5pt, circle, fill=black] {};
  \node  at (67bp,95bp) {\LARGE $y_2$};  
  \node (x_2) at (80bp,130bp) [draw, scale=1.5pt, circle, fill=black] {};
  \node  at (80bp,140bp) {\LARGE $x_2$};
  \node (y_3) at (120bp,50bp) [draw, scale=1.5pt, circle, fill=black] {};
  \node  at (130bp,60bp) {\LARGE $y_3$}; 
  \node (x_3) at (160bp,10bp) [draw, scale=1.5pt, circle, fill=black] {};
  \node  at (170bp,20bp) {\LARGE $x_3$}; 

  \node  at (80bp,-20bp) {\LARGE $H$ and $\rho=\{x_1y_1, x_2y_2, x_3y_3\}$};
  
  \draw [-, line width=1.5pt, gray] (x_1) -- node[above,sloped] {} (y_1);
  \draw [-, dashed] (y_1) -- node[above,sloped] {} (y_2);
  \draw [-, line width=1.5pt, gray] (y_2) -- node[above,sloped] {} (x_2);
  \draw [-, dashed] (y_2) -- node[above,sloped] {} (y_3);
  \draw [-, line width=1.5pt, gray] (y_3) -- node[above,sloped] {} (x_3);

  \node  at (200bp,105bp) {\Large $\Psi$}; 
  \node  at (200bp,90bp) {\Huge $\longrightarrow$}; 

  \fill[pattern color=gray, draw=white, fill=white, pattern=vertical lines] (350bp,-40bp) -- (250bp,50bp) -- (280bp,80bp) --(380bp,-10bp);
  \draw [-, line width = 3pt] (350bp,-40bp) -- node[above,sloped] {} (250bp,50bp);
  \draw [-, line width = 3pt] (380bp,-10bp) -- node[above,sloped] {} (280bp,80bp);

  \node[circle, fill=white]  at (315bp,20bp) {\LARGE $A_{\infty}$};

  \node  at (270bp,90bp) {\LARGE $Y_1$};
  \node  at (240bp,60bp) {\LARGE $X_1$};

  \fill[pattern color=gray, draw=white, fill=white, pattern=horizontal lines] (490bp,-40bp) -- (590bp,50bp) -- (560bp,80bp) --  (460bp,-10bp);
  \draw [-, line width = 3pt] (490bp,-40bp) -- node[above,sloped] {} (590bp,50bp);
  \draw [-, line width = 3pt] (460bp,-10bp) -- node[above,sloped] {} (560bp,80bp);
  \node[circle, fill=white]  at (525bp,20bp) {\LARGE $A_{\infty}$};

  \node  at (450bp,-20bp) {\LARGE $Y_3$};
  \node  at (480bp,-50bp) {\LARGE $X_3$};

 \draw [-,dashed] (510bp,40bp) -- node[above,sloped] {} (440bp,145bp);
 \node[rotate=129,fill=white] (join2) at (475bp,92.5bp) {\LARGE $\join$};

 \draw [-,dashed] (330bp,40bp) -- node[above,sloped] {} (400bp,145bp);
  \node[rotate=50,fill=white] (join1) at (365bp,92.5bp) {\LARGE $\join$};

  \fill[pattern color=gray, draw=white, fill=white, pattern=north west lines] (340bp,190bp) -- (500bp,190bp) -- (500bp,150bp) --  (340bp,150bp);
  \draw [-, line width = 3pt] (340bp,190bp) -- node[above,sloped] {} (500bp,190bp);
  \draw [-, line width = 3pt] (340bp,150bp) -- node[above,sloped] {} (500bp,150bp);

  \node[circle, fill=white]  at (420bp,170bp) {\LARGE $A_{\infty}$};

  \node  at (515bp,190bp) {\LARGE $X_2$};
  \node  at (515bp,150bp) {\LARGE $Y_2$};

  \end{tikzpicture}}
  \end{center}
  \caption{Construction of $\Psi(H,\rho)$ from $(H, \rho)$}
\label{ConstructionOfPsi}
\end{figure}

We now have:

\begin{proposition}\label{Psi(G)}
 For every simple, finite or infinite, graph $H$ with a perfect matching $\rho$,
the graph $\Psi(H,\rho)$ can only be identified with $V(\Psi(H,\rho))$.
\end{proposition}

\begin{proof}
Let $A_x$ be the copy of $A_{\infty}$ which corresponds to the edge $x\rho(x)$. Then for every vertex $y$ in $V(\Psi(H,\rho))\setminus V(A_x)$, either $y$ is connected to every vertex in $A_x$ or to neither of them. Thus to separate vertices in $A_x$, we need all the vertices of $A_x$. Since $x$ is arbitrary, we need  all the vertices in $V(\Psi(H,\rho))$ in any separating set.
\end{proof}

In the next theorem we prove that every such extremal connected infinite graph is
$\Psi(H,\rho)$ for some connected finite or infinite graph $H$ together with a matching $\rho$.

\begin{theorem}\label{ClassifiactionInfinite}
 Let $G$ be an infinite connected graph. Then
a proper subset $C$ of $V(G)$ identifies all pairs of vertices of $G$ unless
$G=\Psi(H,\rho)$ for some finite or infinite graph $H$ together with a perfect matching $\rho$.
\end{theorem}

\begin{proof}
We already have seen that if $G\cong \Psi(H,\rho)$, then the only identifying code of $G$
is $V(G)$. To prove the converse suppose $G-v$ has a pair of twin vertices for every
vertex $v$ of $G$. It is enough to show that every vertex $v$ of $G$ belongs
to a unique induced subgraph $A_v$ of $G$ isomorphic to $A_{\infty}$ and that if
a vertex not in $A_v$ is adjacent to a vertex in the $X$ (respectively, $Y$) part of $A_v$
then it is adjacent to all the vertices of the $X$ (respectively, $Y$).

Let $x_1$ be a vertex of $G$. The subgraph $G-x_1$ has a pair of twins,
let $y_1$ and $y_2$ be one such pair. Assume, without loss of generality, that
$x_1$ is adjacent to $y_2$ and not to $y_1$. By Lemma~\ref{lemma:twins_of_the_twins},
$x_1$ must be one of the vertices of a pair of twins in $G-y_1$.
Let the other be $x_2$. Now consider the subgraph $G-y_1$.
This subgraph must have a pair of twins and $x_1$ must be one of them. Let $x_0$ be the
other one.

Continuing this process in both directions (with negative and positive indices)
we build our $A_{x_1}\cong A_{\infty}$ as a subgraph of $G$.
Since each consecutive pair of vertices in $X\subset A_{x_1}$ is separated only by a vertex
in $Y\subset A_{x_1}$, every pair of vertices in $X$ are twins in $G-Y$. Thus each vertex not in $A_{x_1}$, either is adjacent to all the vertices in $X$ or to none of them. Similarly, every vertex in $A_{x_1}$, either is adjacent to all the vertices in $Y$ or to none.
Hence $A_{x_1}$ is unique. This proves the theorem.
\end{proof}

\section{Bounding $\gamma^{\text{\tiny {ID}}}(G)$ by $n$ and $\Delta$ }\label{BoundBynAndDelta}

In this section, we introduce new upper bounds on parameter
$\gamma^{\text{\tiny {ID}}}$ in terms of both the order and the maximum
degree of graph, thus extending a result of~\cite{GM07}.

We define $A^+_{\infty}$ to be the subgraph of $A_\infty$ induced by the vertices
of positive indices in $X$ and in $Y$.
The following lemma, which is a strengthening of
Theorem~\ref{thm:existence}, has been attributed to N.~Bertrand~\cite{B01}. We give an independent proof as~\cite{B01} is not accessible.

\begin{lemma}[\cite{B01}]\label{lemma:takeout_vertex}
  If $G$ is a twin-free graph (infinite or not) not containing $A^+_{\infty}$ as an induced subgraph,
  then for every vertex $x$ of $G$, there is a vertex $y \in B_1(x)$ such that $G - y$ is twin-free.
\end{lemma}

\begin{proof}
 By contradiction, suppose that $x_1$ is a vertex that fails the statement of the lemma.
Then $G- x_1$ has a pair of twin vertices. We name them $y_1$ and $y_2$.
Without loss of generality we assume that $x_1$ is adjacent to
$y_2$ but not to $y_1$. Now, in $G- y_2 $ we must have another pair $u, u'$ of twin vertices.
By Lemma~\ref{lemma:twins_of_the_twins}, $x_1 \in \{u,u'\}$, we name the other
element $x_2$ ($x_2 \in B_1(x_1)$). Note that the subgraph induced on
$x_1,x_2,y_1,y_2$ is isomorphic to $A_2$. We prove 
by induction that $A^+_{\infty}$ is an induced subgraph of $G$, thus obtaining a contradiction.

To this end suppose $A_k$ on $\{y_1,\ldots y_k,x_1,\ldots, x_k \}$ is already built such that
$x_{k-1}, x_k$ are twins in $G-y_k$ and $y_{k-1}, y_k$ are twins in $G- x_{k-1}$.
Then $x_k \in B_1(x_1)$. Consider $G-x_k$. There must be a pair of twins and,
by Lemma~\ref{lemma:twins_of_the_twins}, $y_k$ must be one of them. 
Let $y_{k+1}$ be the other one. Since $y_k$ and $y_{k+1}$ are twins in $G-x_k$,
 then $y_{k+1}$ is adjacent to $x_1, \ldots, x_k$ and $y_1, \ldots, y_k$, in particular
$y_{k+1}\in B_1(x_1)$. Now, there
must be a pair of twins in $G-y_{k+1}$ and again by Lemma~\ref{lemma:twins_of_the_twins} one of them must be 
$x_k$, let the other one be $x_{k+1}$. Since $x_k$ and $x_{k+1}$ are twins in $G-y_{k+1}$, then
$x_{k+1}$ is adjacent to $x_1, \ldots ,x_k$ and not adjacent to $y_1,\ldots, y_k$.
Thus the graph induced on $\{y_1,\ldots ,y_{k+1},x_1,\ldots, x_{k+1} \}$ is isomorphic
to $A_{k+1}$ with the property that $x_k$,$x_{k+1}$ are twins in $G-y_{k+1}$ and $y_k$,$y_{k+1}$ are twins in  $G-x_k$. 
Since this process does not end, we find that $A^+_{\infty}$ is an induced subgraph of $G$.
\end{proof}

It was conjectured in~\cite{F09} that:

\begin{conjecture}[\cite{F09}]
For every connected twin-free graph $G$ of maximum degree $\Delta\ge
3$, we have $\gamma^{\text{\tiny{ID}}}(G)\leq\Big\lceil|V(G)|-\frac{|V(G)|}{\Delta(G)}\Big\rceil$.
\end{conjecture}

In support of this conjecture, we prove the following weaker upper
bound on the size of a minimum identifying code of a twin-free
graph. We note that a similar bound is proved in~\cite{F09}.

\begin{theorem}\label{BoundingByDelta^5}
Let $G$ be a connected, twin-free graph on $n$ vertices and
of maximum degree $\Delta$. Then 
$\gamma^{\text{\tiny {ID}} }(G)\leq 
n(1-\frac{\Delta-2}{\Delta(\Delta -1)^5-2})= n- \frac{n}{\Theta(\Delta^5)}$.
\end{theorem}

\begin{proof}
  First, we note that if $I$ is a maximal 6-independent set, then
  $|I|\geq \frac{n(\Delta-2)}{\Delta(\Delta-1)^5-2}$. This is true
  because $|B_5(x)|\leq \frac{\Delta(\Delta-1)^5-2}{\Delta-2}$ for
  every vertex $x$. Now, let $I$ be a 6-independent set. For each
  vertex $x\in I$ let $f(x)$ be the vertex found using
  Lemma~\ref{lemma:takeout_vertex} and $f(I)=\{f(x)~|~x \in I\}$.
  Since $I$ is a 6-independent set, $f(I)$ is a 4-independent set of
  $G$ and $|f(I)|=|I|$.  Now, by Lemma~\ref{4-independent}, we know
  that $C=V(G)\setminus f(I)$ is an identifying code of $G$. The bound
  is now obtained by taking any maximal 6-independent set $I$.
\end{proof}

It is easy to observe that if $G$ is a regular twin-free graph, then $V(G)-x$ is an identifying code for every vertex $x$ of $G$. Thus the result of theorem~\ref{BoundingByDelta^5} can be slightly improved for regular graphs as follows:

\begin{theorem}
Let $G$ be a connected $\Delta$-regular twin-free graph on $n$ vertices. 

Then 
$\gamma^{\text{\tiny {ID}} }(G)\leq 
n(1-\frac{1}{1+\Delta-\Delta^2+\Delta^3})= n- \frac{n}{\Theta(\Delta^3)}$.
\end{theorem}

\begin{proof}
We note that a $4$-independent set $I$ of size at least 
$\frac{n}{1+\Delta-\Delta^2+\Delta^3}$ can be found because $|B_3(x)|\leq \frac{\Delta(\Delta-1)^3-2}{\Delta-2}=1+\Delta-\Delta^2+\Delta^3$.
Now, $G-x$ is twin-free for every vertex $x$ of $I$ (because $G$ is regular), so by Lemma~\ref{4-independent}, $V(G)-I$ is an identifying code of $G$.
\end{proof}

It is proved in~\cite{GM07} that in any nontrivial infinite twin-free
graph $G$ whose vertices are all of finite degree, there exists a vertex $x$
such that $V(G)\setminus \{x\}$ is an identifying code of $G$. Using
Lemma~\ref{lemma:takeout_vertex} and similar to the proof of
Theorem~\ref{BoundingByDelta^5}, we can strengthen their result as
follows:

\begin{theorem}\label{InfiniteofInfinite}
  Let $G$ be a connected infinite twin-free graph whose vertices all
  have finite degree. Then there exists an infinite set of vertices
  $I\subseteq V(G)$, such that $V(G)\setminus I$ is an identifying code of $G$.
\end{theorem}

\section{General $r$-identifying codes}\label{r-identifyingCodes}

To identify the class of graphs with $\gamma^{\text{\tiny {ID}}
}_r(G)=n-1$ one needs to find the $r$-roots of the graphs in
$\{K_{1,t}~|~t\geq 2\}\}\cup \mathcal A \cup (\mathcal A \join K_1)$.  The
general problem of finding the $r$-root of a graph $H$ is an NP-hard
problem~\cite{MS94} and it does not seem to be an easy task in this
particular case either.

If $s$ divides $k-1$ and $r=\frac{k-1}{s}$, then the graph $G=P_{2k}^s$ is one of the
$r$-roots of $A_k$. It is easy to see that, in most cases, one can remove many edges
of $G$ and still have $G^r\cong A_k$. The difficulty of the problem is that an $r$-root
of $A_k$ is not necessarily a subgraph of $P_{2k}^s$. An example of such a 2-root of $A_5$
is given in Figure~\ref{powerPath}.

\begin{figure}[!ht]
\centering
\begin{tikzpicture}[join=bevel,inner sep=0.5mm,]

  \node (1) at (90bp,110bp) [draw, circle, fill=black] {};
  \node  at (90bp,100bp) {$1$};
  
  \node (2) at (130bp,110bp) [draw, circle, fill=black] {};
  \node  at (130bp,100bp) {$2$};
  
  \node (3) at (170bp,110bp) [draw, circle, fill=black] {};
  \node  at (170bp,100bp) {$3$};
    
  \node (4) at (210bp,110bp) [draw, circle, fill=black] {};
  \node  at (210bp,100bp) {$4$};

  \node (5) at (250bp,110bp) [draw, circle, fill=black] {};
  \node  at (250bp,100bp) {$5$};
  
  \node (6) at (290bp,110bp) [draw, circle, fill=black] {};
  \node  at (290bp,100bp) {$6$};
  
  \node (7) at (330bp,110bp) [draw, circle, fill=black] {};
  \node  at (330bp,100bp) {$7$};
  
  \node (8) at (370bp,110bp) [draw, circle, fill=black] {};
  \node  at (370bp,100bp) {$8$};
  
  \node (9) at (410bp,110bp) [draw, circle, fill=black] {};
  \node  at (410bp,100bp) {$9$};
  
  \node (10) at (450bp,110bp) [draw, circle, fill=black] {};
  \node  at (450bp,100bp) {$10$};
  
  \draw [-] (1) -- node[above,sloped] {} (2);
  \draw [-] (2) -- node[above,sloped] {} (3);
  \draw [-] (3) -- node[above,sloped] {} (4);
  \draw [-] (4) -- node[above,sloped] {} (5);
  \draw [-] (5) -- node[above,sloped] {} (6);
  \draw [-] (6) -- node[above,sloped] {} (7);
  \draw [-] (7) -- node[above,sloped] {} (8);
  \draw [-] (8) -- node[above,sloped] {} (9);
  \draw [-] (9) -- node[above,sloped] {} (10);
  
\draw [semithick]  
   (90bp,110bp) to [out=40,in=140] (170bp,110bp)
   (130bp,110bp) to [out=40,in=140] (250bp,110bp)
   (170bp,110bp) to [out=30,in=150] (250bp,110bp)
   (210bp,110bp) to [out=40,in=140] (330bp,110bp)
   (290bp,110bp) to [out=40,in=140] (410bp,110bp)
   (290bp,110bp) to [out=30,in=150] (370bp,110bp)
   (370bp,110bp) to [out=40,in=140] (450bp,110bp);

\end{tikzpicture}
\caption{A 2-root of $A_5$ which is not a subgraph of $P_{10}^2$ }
\label{powerPath}
\end{figure}
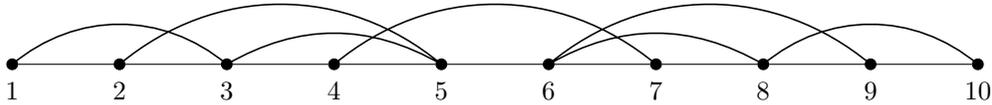

For the case of infinite graphs, we note that there exists a 2-root of
$A_{\infty}$. This graph is defined as follows: it has the same vertex
set $X \cup Y$ as $A_{\infty}$ and the same edges between $X$ and $Y$,
but no edges within $X$ or $Y$. However, we do not know whether there
exist other roots of graphs described in
Theorem~\ref{ClassifiactionInfinite}.

We should also note that a $(3r+1)$-independent set in $G^r$ is a
4-independent set in $G$. Thus we have the following general form of
Lemma~\ref{4-independent}, Theorem~\ref{BoundingByDelta^5} and
Theorem~\ref{InfiniteofInfinite}:

\begin{lemma}
 Let $G$ be a connected graph on $n$ vertices such
 that $G^r$ is twin-free. Let $I$
 be a $(3r+1)$-independent set of $G$ such that for every vertex $v$
 of $I$ the set $V(G)\setminus \{v\}$ is an $r$-identifying code of $G$. Then
 $C=V(G)\setminus I$ is an $r$-identifying code of $G$.
\end{lemma}

\begin{theorem}
 Let $G$ be a connected graph on $n$ vertices and
of maximum degree $\Delta$ such that $G^r$ is twin-free. Then 
$\gamma_r^{\text{\tiny {ID}} }(G)\leq 
n(1- \frac{\Delta-2}{\Delta(\Delta -1)^{5r}-2}) = n- \frac{n}{\Theta(\Delta^{5r})}$.
\end{theorem}

\begin{theorem}
  Let $G$ be a connected infinite graph whose vertices are of finite
  degree such that $G^r$ is twin-free. Then there exists an infinite set of vertices
  $I\subseteq V(G)$, such that $V(G)\setminus I$ is an $r$-identifying code of $G$.
\end{theorem}

\section{Remarks}

We conclude our paper by some remarks on related works.

{\bf Remark 1} 
The following two questions were posed in~\cite{S07}:
\begin{enumerate}
\item Do there exist $k$-regular graphs $G$ of order $n$ with
  $\gamma^{\text{\tiny {ID}}}(G)=n-1$ for $k<n-2$?
\item Do there exist graphs $G$ of odd order $n$ and maximum
  degree~$\Delta < n-1$ with $\gamma^{\text{\tiny {ID}}}(G)=n-1$?
\end{enumerate}

As a corollary of Theorem~\ref{thm:classification}, we can now answer
these questions in the negative. Indeed, for the first question, if $G$ is 
a $k$-regular ($k\geq 2$) graph of order $n$ with
$\gamma^{\text{\tiny {ID}}}(G)=n-1$ then $G$ 
is the join of $k$ disjoint copies of $A_1$.
For the second question,  noting that each graph in $\mathcal A$ has an even order,
we conclude that if a graph $G$ on an odd number, $n$, of vertices has 
$\gamma^{\text{\tiny {ID}}}(G)=n-1$, then 
$G\in \{K_{1,t}~|~t\geq 2 \} \cup (\mathcal A \join K_1)$ and, therefore $\Delta(G)=n-1$.

{\bf Remark 2}
Given a graph $G=(V,E)$ the \emph{1-ball membership graph} of $G$ is defined to be
the bipartite graph $G^*=(I\cup A, E^*)$ where $I=V(G), A=\{B_1(x) ~|~ x \in V(G)\}$
and $E^*=\{\{u,B_1(v)\}~|~ u\in B_1(v), u,v \in V(G)\}$.
It is not hard to see that the problem of finding identifying codes in $G$ is equivalent
to the one of finding discriminating codes in $G^*$. But since not every bipartite graph
is a 1-ball membership graph, the latter contains the former properly. It is a rephrasing
of Bondy's theorem~\cite{B72}, that every bipartite graph $(I\cup A, E)$ has a discriminating
code of size at most $|I|$. The class of bipartite graphs $(I\cup A, E)$ in which any 
discriminating code has size at least $|I|$ are classified in~\cite{CCCHL08}.  They
further asked for the classification of bipartite graphs in which every discriminating code
needs at least $|I|-1$ vertices of $A$.
In Theorem~\ref{thm:classification} we answered this question for those bipartite 
graphs that are isomorphic to a 
1-ball membership of a graph.

{\bf Acknowledgments}\\ We would like to acknowledge S.~Gravier, R.~Klasing and 
A.~Kosowski for helpful discussions on the topic of this paper. The example of Figure~\ref{powerPath} was found during a discussion with A.~Kosowski. 
We also would like to thank the referee for careful reading and for helping us with a better presentation.

\end{document}